\documentclass[11pt]{amsart}

\usepackage{amsmath}
\usepackage{amssymb}
\usepackage{amsthm}
\usepackage{graphicx}
\usepackage{tikz}
\usepackage{xcolor}

\allowdisplaybreaks

\usepackage[margin=1.2in]{geometry}
\def\E{\mathbb{E}}

\newcommand{\bl}{\boldsymbol{\lambda}}

\newcommand{\blambda}{\boldsymbol{\lambda}}

\newcommand{\dist}{\mathrm{dist}}

\newcommand{\one}{\mathbf{1}}

\newcommand{\bY}{\mathbf{Y}}
\newcommand{\bz}{\mathbf{z}}

\def\E{\mathbb{E}}
\def\R{\mathbb{R}}

\def\C{\mathbb{C}}
\def\Z{\mathbb{Z}}

\def\eps{\varepsilon}
\def\del{\delta}

\def\cA{\mathcal{A}}

\def\brho{\boldsymbol {\rho}}

\def\1{\mathbf{1}}

\def\lam {\lambda}
\def\Lam{\Lambda}
\def\tce{t_c + \eps}
\def\tce2{t_c + \frac{\eps}{2}}

\def\bX{\mathbf{X}}

\newtheorem*{theorem*}{Theorem}
\newtheorem{theorem}{Theorem}
\newtheorem{lemma}[theorem]{Lemma}
\newtheorem{cor}[theorem]{Corollary}
\newtheorem{defn}[theorem]{Definition}
\newtheorem*{defn*}{Definition}

\newtheorem*{prop*}{Proposition}

\newtheorem*{conj*}{Conjecture}

\newtheorem*{fact*}{Fact}
\newtheorem{fact}[theorem]{Fact}

\theoremstyle{remark}

\newtheorem*{remark*}{Remark}

\begin{document}
	\title{Analyticity for classical gasses  via recursion}
	\author{Marcus Michelen}
	\address{Department of Mathematics, Statistics, and Computer Science\\ University of Illinois at Chicago}
	\author{Will Perkins}
	\address{School of Computer Science\\ Georgia Institute of Technology}
	\email{michelen.math@gmail.com \\ math@willperkins.org}
	\date{\today}
	
	\begin{abstract}
		We give a new criterion for a classical gas with a repulsive pair potential to exhibit uniqueness of the infinite volume Gibbs measure and analyticity of the pressure.   Our improvement on the bound for analyticity is by a factor $e^2$ over the classical cluster expansion approach and a factor $e$ over the known limit of cluster expansion convergence.    The criterion is based on a contractive property of a recursive computation of the density of a point process.  The key ingredients in our proofs include an integral identity for the density of a Gibbs point process and an adaptation of the algorithmic  correlation decay method from theoretical computer science.  We also deduce from our results an improved bound for analyticity of the pressure as a function of the density.   
	\end{abstract}
	
	\maketitle

	\section{Introduction}
	\label{secIntro}

	A central goal in classical statistical mechanics is to derive macroscopic properties of  gasses, liquids, and solids given only details of microscopic interactions.  The classical model involves indistinguishable particles interacting in the continuum via a many-body potential.  Even in the case of a pair potential fundamental questions remain open, such as whether or not such a model exhibits a phase transition from a gaseous to a solid state.  One major success of the field is in the study of convergent expansions for thermodynamic quantities in powers of the activity or density; using these techniques, the model is proved to be in a gaseous state throughout the region of convergence of these series.  However, convergence can be limited by the presence of non-physical singularities in the complex plane.  
	Here, in the case of repulsive pair potentials, we extend the known range of the gaseous phase using analytic techniques that avoid such singularities and focus on small complex neighborhoods of the positive real axis.

	\subsection{The model}

	We consider classical particles interacting in a finite volume $\Lam \subset \R^d$ via a symmetric, translation invariant pair potential $\phi$.   
	
	The energy of a configuration $\{x_1, \dots, x_n \} \subset \R^d$ is given by
	\begin{align*}
	U(x_1, \dots , x_n ) &= \sum_{1 \le i < j \le n } \phi(x_i - x_j) \,.
	\end{align*}
	We make the following assumptions on the potential $\phi : \R^d \to \mathbb R \cup \{ +\infty \}$:
	\begin{enumerate}
		\item $\phi$ is \textit{repulsive}: $\phi(x) \ge 0$ for all $x$.
		\item $\phi$ is \textit{tempered}:  $\int_{\R^d}  |1- e^{- \phi(x)}| \, dx < \infty$.
	\end{enumerate}
	
	The grand canonical partition function at activity $\lam \ge0$ on a bounded region $\Lam \subset \R^d$ is given by:
	\begin{equation}
	\label{eqGPPParition}
	Z_\Lam( \lam) = 1+ \sum_{ k \ge 1} \frac{\lam ^k}{k!} \int_\Lam \cdots \int_\Lam e^{- U(x_1, \dots x_k)}   \,dx_1 \cdots \, dx_k \,. 
	\end{equation}
	(Without losing any generality, we absorb the usual inverse temperature parameter $\beta$ into the potential $\phi$).

	If we let $\Lam_V$ be the axis-parallel box of volume $V$ centered at the origin in $\R^d$, then the infinite volume pressure is 
	\begin{align*}
	p_\phi(\lam) = \lim_{V \to \infty} \frac{1}{V} \log Z_{\Lam_V}( \lam)  \,.
	\end{align*}
	See e.g.~\cite{ruelle1999statistical} for a proof of the existence of the limit.  Non-analytic points of $p_\phi(\lam)$ on the positive real axis mark phase transitions of the infinite volume system.  An important topic in classical statistical physics is proving the absence of phase transitions for certain  parameter values; that is, proving analyticity of the pressure.  
	We remark that proving the existence of a phase transition in the models considered here is a notoriously challenging problem; to the best of our knowledge it is not proved for any monatomic gas interacting through a finite-range or rapidly decaying pair potential (see~\cite{widom1970new,lebowitz1999liquid} for proofs of phase transitions in continuous particle systems with different types of interactions).    For more background on classical gasses see~\cite{ruelle1999statistical}.

	\subsection{Main results}
	
	Our main result is a new criterion for analyticity of the pressure and uniqueness of the infinite volume Gibbs measure.  
	\begin{theorem}
		\label{thmAnalytic}
		Let $\phi$ be a repulsive, tempered potential and let 
		\begin{equation}
		\label{eqCPhiRepulsive}
		C_{\phi} = \int_{\R^d}  \left |1- e^{- \phi(x)}  \right|   \, dx  \,.
		\end{equation}
		Then for $\lam \in [0, e/ C_{\phi})$ the infinite volume pressure $p_\phi(\lam)$ is analytic and there is a unique infinite volume Gibbs measure.
	\end{theorem}
	
	This improves by a factor $e^2$ the classical bound of $1/(e C_{\phi})$ obtained by Groeneveld in 1962 using the cluster expansion~\cite{groeneveld1962two}.  Extensions by Penrose~\cite{penrose1963convergence} and Ruelle~\cite{ruelle1963correlation} a year later to a wider class of potentials via the Kirkwood--Salsburg equations  match Groeneveld's bound in the case of repulsive potentials.  In a remarkable but little noticed paper, Meeron~\cite{meeron1970bounds} proved uniqueness and analyticity for $\lam< 1/C_{\phi}$ using a novel interpolation scheme~\cite{meeron1962indirect} and recurrence relations for $k$-point densities related to the Kirkwood--Salsburg equations.   The bound $1/C_{\phi}$ has since been matched by other methods, including Dobrushin uniqueness (for potentials with a hard-core)~\cite{houdebert2020explicit}  and disagreement percolation~\cite{betsch2021uniqueness}.
	
New criteria for cluster expansion convergence have been given by Faris~\cite{faris2008connected}, Jansen~\cite{jansen2019cluster}, and Nguyen and Fern{\'a}ndez~\cite{nguyen2020convergence}  based on the work of Fern{\'a}ndez and Procacci in discrete setting~\cite{fernandez2007cluster} and the extension by Fern{\'a}ndez, Procacci, and Scoppola to hard spheres~\cite{fernandez2007analyticity}.  Explicit improvements to the classical bounds from the new criteria  have been worked out in the case of hard spheres in dimension $2$ (see Section~\ref{secHardSphere} below).   
	
	Moreover, $1/ C_\phi$ is an upper bound on the radius of convergence of the cluster expansion for repulsive potentials~\cite{groeneveld1962two,penrose1963convergence} (see remark 3.7 in~\cite{nguyen2020convergence}).  Theorem~\ref{thmAnalytic} surpasses this limit (and Meeron's bound) by a factor $e$.

	Using Theorem~\ref{thmAnalytic} we can also deduce that the pressure is analytic with respect to the infinite volume density defined by 
	$$\rho_{\phi}(\lambda) := \lam \frac{d}{d \lambda} p_{\phi}(\lambda)\,.$$
	\begin{theorem}
		\label{corDensity}
		For any repulsive, tempered potential $\phi$, the infinite volume pressure $p_{\phi}$ is analytic as a function of the density $\rho_{\phi}$ for $ \rho_{\phi}  \in [0,\frac{e}{1+e} \frac{1}{C_{\phi}} )$.  
	\end{theorem}
	Previous results showed analyticity of the pressure as a function of the density by showing convergence of the virial series in a disk around $0$ in the complex plane~\cite{lebowitz1964convergence,pulvirenti2012cluster,jansen2019virial,nguyen2020convergence}.  The best bound on convergence of the virial expansion for general repulsive potentials is that of Groeneveld~\cite{groeneveld1967estimation}, showing convergence for $|\rho_{\phi}| \le .237961 \frac{1}{C_{\phi}}$ (see also~\cite{ramawadh2015virial}).

	\subsubsection{Example: the hard sphere model}
	\label{secHardSphere}
	One of the most studied classical gasses is the hard sphere model, with the potential $\phi(x) = +\infty$ if $\|x \| < r $ and $0$ otherwise.  A configuration $(x_1, \dots, x_k)$ with $U(x_1, \dots, x_k) =0$ represents the centers of a packing of $k$ spheres of radius $r/2$.  This model provides a good testing ground for different criteria for uniqueness and analyticity.   
	
	For convenience, let us take $r$ to be the radius of the ball of volume $1$ in $\R^d$.  Then the classical Groeneveld--Penrose--Ruelle bound gives convergence of the cluster expansion and  uniqueness for $\lam < 1/e \approx .3679$.   Fern\'andez, Procacci, and Scoppola~\cite{fernandez2007analyticity} gave an improved bound in terms of multi-dimensional integrals.  For $d=2$ their improved bound is $\lam < .5017$.  
	
	Hofer-Temmel~\cite{christoph2019disagreement} and Dereudre~\cite{dereudre2019introduction} gave a new bound for uniqueness and exponential decay of correlations for the hard sphere model based on disagreement percolation and the critical activity for Poisson-Boolean percolation.  In high dimensions this gives uniqueness for $\lam < 1 +o_d(1)$.  In dimension $2$, using the `high-confidence' bound for Poisson-Boolean percolation from~\cite{balister2005continuum}  using Monte-Carlo integration, this approach gives uniqueness for  $\lam <4.508$, a significant improvement over the cluster expansion approach.   The best-known rigorous bound for the critical value in Poisson-Boolean percolation appears to be that given in~\cite[Theorem 3.10]{meester1996continuum} translating to uniqueness for $\lam < 4\pi \cdot .174 \approx 2.187$. This approach also works for repulsive potentials with finite range~\cite{dereudre2019introduction,benevs2019decorrelation}.  Subsequent work of Betsch and Last  extended the disagreement percolation method to repulsive potentials with infinite range~\cite{betsch2021uniqueness}; the bound obtained gives uniqueness for $\lam < 1/C_{\phi}$, with better bounds possible in specific cases based on continuum percolation bounds. 
	
	Recently, Helmuth, Petti, and the second author~\cite{helmuth2020correlation} improved these bounds in dimensions $d \ge 3$  by showing uniqueness and exponential decay of correlations in the hard sphere model for $\lam < 2$ using path coupling to show rapid mixing of a discrete-time Markov chain. 
	
	Theorem~\ref{thmAnalytic} implies a further improvement to these bounds.
	\begin{cor}
		\label{corHardSphere}
		Consider the hard sphere model with spheres of radius $r/2$ where $r$ is the radius of the ball of volume $1$ in $\R^d$.  Then for $\lam < e$ the pressure is analytic and there is a unique infinite-volume Gibbs measure.
	\end{cor}
	
	Theorem~\ref{corDensity} also gives an improvement on the bounds for analyticity of the pressure as a function of the density of the hard sphere model.  We write these bounds in terms of the packing density of the model, which given our choice of normalization is $2^{-d} \rho_{\phi}(\lam)$ since the volume of the ball of radius $r/2$ is $2^{-d}$. 
	
	\begin{cor}
		\label{corHSdensity}
		The infinite volume pressure of the hard sphere model in dimension $d$ is analytic as a function of the density for packing densities up to $ \frac{e}{1+e} 2^{-d}$.  Moreover, for any $\eps>0$ and $d$ large enough, the pressure of the hard sphere model in dimension $d$ is analytic for packing densities up to $(1-\eps) 2^{-d}$.  
	\end{cor}
	For instance, in dimension $2$ this gives analyticity up to packing density $.18276$.  The best previous lower bound on the critical density (in terms of uniqueness of the infinite volume Gibbs measure) was $1/6$ from~\cite{helmuth2020correlation}.  The best bound obtained via convergence of the virial expansion is $.0751$ by Nguyen and Fern{\'a}ndez~\cite{nguyen2020convergence}.  The new bound is still far from the predicted critical packing density near $.7$~\cite{bernard2011two,engel2013hard}.  
	
	Packing density $2^{-d}$ is a natural barrier to analysis since below $2^{-d}$ free volume (space in which new centers can be placed) is guaranteed by a union bound.  Improving the asymptotic bound in Corollary~\ref{corHSdensity} by any constant factor would likely require significant new insight. 
	
	\subsection{Absence of zeros in the complex plane}
	
	To prove Theorem~\ref{thmAnalytic} we will work in the setting of a multivariate, complex-valued partition function.  Given any bounded, measurable function $\bl : \Lam \to \mathbb C$, we define   
	\begin{equation}
	Z_\Lam( \bl) = \sum_{k \ge 0} \frac{1}{k!} \int_{\Lambda^k} \prod_{i=1}^k \bl(x_i)  \cdot  e^{- U(x_1, \dots, x_k) } d x_1 \cdots dx_k  \,.
	\end{equation}
	When $\bl $ is constant this definition coincides with~\eqref{eqGPPParition}. 
	
	Following the Lee-Yang theory of phase transitions~\cite{yang1952statistical}, absence of phase transitions and analyticity of the pressure for activities in $[0,\lam_0]$ is closely related to the existence of a region $R$ in the complex plane containing the segment $[0,\lam_0]$ so that for any $\lam \in R$ and any bounded region $\Lam$, $Z_\Lam(\lam) \ne 0$.  Theorem~\ref{thmAnalytic} thus follows, after a little complex analysis in Section~\ref{secFinalProofs}, from the next theorem.  
	
	\begin{theorem}
		\label{thmZeros}
		Let $\phi$ be a repulsive, tempered potential and  suppose $\lam \in (0, e/C_{\phi})$.  Then there exist $\eps>0$ and $C>0$, so that the following holds.  Let $L_0$ be the $\eps$-neighborhood of the interval $[0, \lam]$ in the complex plane.  Then for any  measurable function $\bl: \R^d \to \mathbb C$ so that $\bl(x) \in L_0$ for all $x$, and any bounded, measurable $\Lam \subset \R^d$, we have 
		\[|\log Z_\Lam( \bl)| \leq  C |\Lam| \, .\]
		In particular, $ Z_\Lam( \bl) \ne 0$. 
	\end{theorem}

	\subsection{Outline of techniques}
	\label{secTechniques}
	
	The classical approach to proving absence of phase transitions, convergence of the cluster expansion, is limited by the possible presence of a  complex zero of the partition function far from the positive real axis which determines the radius of convergence of the cluster expansion but does not affect the physical behavior of the system.  In fact, for repulsive gasses, the closest singularity of the pressure to the origin lies on the negative real axis~\cite{groeneveld1962two}. Therefore to move beyond the limits of cluster expansion convergence one must use properties of positive activities or utilize regions of the complex plane that are not symmetric around $0$.   Two previous approaches in this direction are probabilistic: disagreement percolation~\cite{dereudre2019introduction,christoph2019disagreement,benevs2019decorrelation} and Markov chain mixing~\cite{helmuth2020correlation}.  In the case of the hard sphere model these techniques surpass the bounds for uniqueness given by the cluster expansion.  One drawback is that these arguments rely on a finite-range property of a potential (like hard spheres) and it is not clear how to extend these arguments to potentials satisfying the more natural temperedness assumption.  
	
	Our approach to proving Theorem~\ref{thmAnalytic} will instead be analytic, using essentially no probabilistic tools.  One consequence of the convergence of the cluster expansion is that $Z_\Lam(\lam)$ is not zero for $\lam$ in a disk in the complex plane, uniformly in the region $\Lam$.  To avoid the singularity on the negative real axis, we instead prove that $Z_\Lam(\lam)$ is not zero in an asymmetric region in the complex plane that contains $[0, e/C_{\phi})$.   We obtain this zero-free region by proving recursive bounds on the log partition function and on a generalization of the density of a point process to complex parameters. 
	
	Our approach is inspired by the `correlation decay method', an algorithmic technique from theoretical computer science for approximating the partition function of a discrete spin model.  The method was introduced by Weitz~\cite{Weitz} in an influential paper on approximate counting and sampling from the hard-core model on a graph $G$.

	At the core of Weitz's argument is a formula for the marginal probability that a given vertex $v$ is occupied in the independent set drawn from the hard-core model on $G$ at activity $\lam$.   With $R_{G,v}$ denoting the ratio of the occupation probability of $v$ in $G$ to the non-occupation probability, the formula becomes
	\begin{equation}
	\label{eqWeitzRecurs}
	 R_{G,v} = \lam \prod_{i=1}^d \frac{1}{1+ R_{G_i,u_i}} \, ,
	 \end{equation}
	where $u_1, \dots, u_{d}$ are the neighbors of $v$ and $G_i$ is a graph obtained from $G$ by removing $v$ and some other specified vertices.  The identity \eqref{eqWeitzRecurs} follows by writing the occupation probability as a ratio of partition functions and then writing this ratio as a telescoping product.  One can repeat the calculation in~\eqref{eqWeitzRecurs} recursively, applying it to calculate $R_{G_i,u_i}$ for each neighbor $u_i$ in terms of its neighbors in $G_i$.  In each step of the recursion more vertices are removed from the graphs until one ends up computing the occupation ratio for a graph consisting of a single vertex, a trivial calculation.
	
	There are two equivalent  interpretations of this recursive sequence of calculations.   The first is that of a `computational tree', in which nodes represent computations of occupation ratios in specified graphs, organized in a hierarchical fashion in which the graph associated to node $w$ in the tree is a subgraph of the graph associated to node $u$ if $w$ is a descendant of $u$.    The second interpretation is that the  sequence of calculations constructs a `self-avoiding walk tree' $T$ with root $r$ so that $R_{G,v} = R_{T,r}$, so that a recursive procedure for computing the occupation probability of the root of a tree can be applied to compute the occupation probability of a given vertex in a general graph.  Such a construction had appeared previously in the context of matchings in graphs in~\cite{godsil1981matchings} and implicitly in the work of \cite{scott2005repulsive}.    
	
Weitz then showed that if the original graph $G$ has maximum degree $\Delta$ and the hard-core model on the infinite $\Delta$-regular tree exhibits weak spatial mixing, then the computational tree (and self-avoiding walk tree) exhibits exponential decay of dependence on the values of the  computation at depth $t$.  In particular this means that the computational tree can be truncated and still output a good approximation of $R_{G,v}$.   The threshold for weak spatial mixing on the $\Delta$-regular tree can be explicitly determined via fixed point equations as $\lam_c(\Delta) = \frac{(\Delta-1)^{\Delta-1}}{(\Delta-2)^{\Delta}}$.  Consequences include an efficient algorithm for approximating the partition function on any graph of maximum degree $\Delta$ for $\lam < \lam_c(\Delta)$, and a lower bound on the uniqueness threshold of the hard-core model on $\Z^d$ of $\lam_c(2d)$.  The correlation decay method has since been refined and used to prove the best current lower bounds on the uniqueness threshold for the hard-core lattice gas model on $\Z^2$~\cite{vera2015improved,sinclair2017spatial}.
	
	The extension of the method most relevant for our approach is the paper of Peters and Regts~\cite{peters2019conjecture} in which they use the recursion of Weitz applied to complex values of $\lam$ along with ideas from complex dynamical systems to prove the existence of a zero-free region for the independence polynomial of graphs of maximum degree $\Delta$ in a complex neighborhood containing $[0, \lam_c(\Delta))$, thus solving a conjecture of Sokal~\cite{sokal2000personal}.  Their work was in part motivated by another approach to approximate counting, the polynomial interpolation method of Barvinok~\cite{barvinok2016combinatorics}.

	Our goal in this paper is to develop a version of the correlation decay method for continuous particle systems.  This presents several challenges including determining a useful analogue of the  recurrence given by~\eqref{eqWeitzRecurs}.  We briefly summarize our approach here, with precise definitions postponed until later.  Our main object of study is the density function of a point process, $\rho_{\bl}(v)$, and more generally  a generalization of the density function to complex-valued activity functions (given in~\eqref{eq1ptDef} below and used already in~\cite{ruelle1963correlation}). 

We develop several  tools to work with these complex densities.  The first is Theorem~\ref{thmMainIdentitiy} which provides an  integral identity for $\rho_{\bl}(v)$ in terms of densities with respect to modified activity functions $\bl_{v \to w}$.
\begin{equation}
\label{eqIntIntro}
		\rho_{\bl}(v) = \bl (v) \exp\left(-\int_{\R^d}\rho_{\bl_{v \to w}}(w)(1 - e^{-\phi(v-w)})     \,dw  \right)\,.
		\end{equation}
  This is inspired by the identity~\eqref{eqWeitzRecurs} from the discrete setting, though here instead of vertices being removed from $G$ to form the graphs $G_i$, we have the original activity function $\bl$ being `damped' (multiplied pointwise  by values in $[0,1]$) to form  the new activity functions $\bl_{v \to w}$. 

The second tool is  Lemma~\ref{lemContract}, which provides  a contractive property of the functional that defines the identity in~\eqref{eqIntIntro}.  In particular, if the densities in the integrand in~\eqref{eqIntIntro} all lie in a specified complex domain, then $\rho_{\bl}(v)$ also does.  The condition for this property determines our convergence criterion.

Finally we  use  an integral expression for $\log Z_\Lam(\lam)$ in terms of densities given in Lemma~\ref{lemZIntegral}.  Along with  Lemma~\ref{lemContract} this allows us to bound $\log  Z_\Lam(\lam)$ and prove the absence of zeroes in a region in the complex plane.  While in the discrete setting,  the proof of zero freeness in~\cite{peters2019conjecture} uses induction on the number of vertices in a graph facilitated by the identity~\eqref{eqWeitzRecurs}, our argument (the proof of Theorem~\ref{thmZeroFreeProof}) uses a kind of continuous induction over activity functions $\bl$ starting with the identically $0$ activity function, facilitated by our identity~\eqref{eqIntIntro} which involves damping activity functions.

	At a very high level, our approach has some similarity to the approach of Penrose and Ruelle via the Kirkwood-Salsburg equations and that of Meeron~\cite{meeron1970bounds}: we write identities involving densities and show that for a certain range of complex parameters the operator defining these identities is contractive in a suitable sense.  In Ruelle's argument, uniqueness follows from invertibility of $1 - \lambda K$ where $K$ is an operator on a Banach space.  When the norm of $\lambda K$ is strictly less than $1$ then invertibility, and thus uniqueness, follows.  Without a deeper understanding of the spectrum of this operator $K$, this approach inherently only provides uniqueness for $\lambda$ in a disk centered at $0$.  Conversely, we work entirely with values of $\lambda$ near the positive interval $[0,e/C_\phi)$ which allows us to avoid the non-physical obstructions to analyticity on the negative real axis.    Meeron's interpolation between zero interaction and full interaction also avoids the singularity on the negative real axis, but the operator defining his recursion is contractive only for $\lam < 1/ C_{\phi}$; it would be interesting to see if combining his approach with a change of coordinates as in Section~\ref{secComplexContraction} below could match the bound of Theorem~\ref{thmAnalytic}.   One could view successive applications of the identity in~\eqref{eqIntIntro} as a non-uniform interpolation from activity $0$ to activity $\lam$.

	\subsection{Future directions}
	
	While we work here only with repulsive potentials, we would be very interested to see if the results could be generalized to the class of stable, tempered potentials, the setting of the results of Penrose~\cite{penrose1963convergence} and Ruelle~\cite{ruelle1963correlation} which include more physically relevant examples such as the Lennard--Jones potential.  For a discussion from the physics perspective of the utility of purely repulsive potentials such as the hard sphere mode, see~\cite{widom1967intermolecular,barker1976liquid}. 
	
	When specialized to real parameters, the functional contraction properties given in Section~\ref{secComplexContraction} can be used to show that a recursive computation of the density exhibits exponentially small dependence on boundary conditions in the height of the recursion.  This is turn can be used to show strong spatial mixing for finite-range, repulsive potentials.  In light of these properties and the algorithmic roots of our techniques, it would be interesting to use these methods to design efficient algorithms for estimating the pressure or density when $\lam < e/ C_{\phi}$.
	
	We believe that one could take advantage of the geometry of  low dimensional Euclidean space (e.g. $d=2,3$) and an appropriate notion of the connective constant of $\R^d$ to improve the bounds of Theorem~\ref{thmAnalytic}, as was done for the hard-core lattice gas by Sinclair, Srivastava, {\v{S}}tefankovi{\v{c}}, and Yin~\cite{sinclair2017spatial}.  Doing this for the hard sphere model in dimension $2$ would test the limits of this method and the analogy with discrete spin systems.

	\subsection{Notation}
	
	A \textit{region} is a bounded, measurable subset of $\R^d$.  The volume of $\Lam \subset \R^d$ is denoted $|\Lam|$.  We denote by $\mathrm{dist}(x,y)$ the Euclidean distance between $x, y \in \R^d$.  An \textit{activity function} on a region $\Lam$ is a bounded, measurable function from $\Lam \to \mathbb C$.  We use a bold symbol, e.g. $\bl$, for a non-constant activity function.  Throughout, a \emph{complex neighborhood} is a bounded open subset of $\C$; in each case here, the complex neighborhoods in question are conformal images of the open unit disk, and thus are simply connected as well.  For  positive parameters $\lambda$ and $\eps$, we write $\mathcal{N}(\lambda,\eps)$ to be the complex neighborhood $\{z \in \C :  \dist(z,[0,\lambda]) < \eps\}$.  For a set $A$ we write $\one_A$ to be the indicator function of the set $A$.
	
	\section{Generalized densities}
	\label{secDensities}
	
	To prove Theorem~\ref{thmZeros}, we use a generalization of the density of a point process to complex activity functions:  the density of $v \in \Lam$ at activity $\bl$ is 
	\begin{equation}
	\label{eq1ptDef}
	\rho_{\Lam,\bl}(v) = \bl (v)  \frac{Z_\Lam (\bl e^{-\phi(v- \cdot)}  ) }{ Z_\Lam(\bl)  }
	\end{equation}
	if  $Z _\Lam(\bl) \ne 0$.   If $Z_\Lam(\bl) = 0$ then the density is undefined.   Here $\bl e^{-\phi(v- \cdot)}  : \Lam \to \mathbb C $ denotes the function $\bl(x) e^{-\phi(v-x)}$.  
	
	In Section~\ref{secFinalProofs} we will also work with multipoint densities: the $k$-point density (or $k$-point correlation function) of $(v_1,\ldots,v_k) \in \Lam^k$ at activity $\bl$ is  
	\begin{equation}
	\label{eqKpt}
	\rho_{\Lam,\bl}(v_1,\ldots,v_k) =\bl(v_1)\cdots \bl(v_k) \frac{Z_\Lam(\bl e^{-  \sum_{i=1}^k\phi(v_i - \cdot)})}{Z_\Lam(\bl)} e^{-U(v_1,\ldots,v_k)}\,.
	\end{equation}
	
	\begin{remark*}
		There is a natural probabilistic interpretation of the densities when $\bl$ is non-negative.  First define the \emph{Gibbs point process} $\bX$ to be the random point set in $\Lambda$ with density proportional to $e^{-U(\cdot)}$ against the Poisson process with intensity $\bl$.  Then the density $\rho_{\Lam,\bl}$ is the density of the measure that assigns to a set $A$ the mass $\E[|\bX \cap A|  ]$ with respect to Lebesgue measure; a similar fact holds for the multipoint density with $\E[|\bX\cap A|]$ replaced by a certain factorial moment.  When $\bl$ is non-negative, these may be taken as a definition of the density $\rho_{\Lam,\bl}$ and \eqref{eq1ptDef} and \eqref{eqKpt} become identities to check (see, e.g.\ \cite[Chapter 4]{ruelle1999statistical}).  Since we are interested in complex values of $\bl$, we use \eqref{eq1ptDef} and \eqref{eqKpt} as our definition.  This idea appears in Ruelle's classic text on statistical mechanics \cite[Chapter 4]{ruelle1999statistical} as well, in which this identity appears in a slightly different form to extend the definition of $\rho_{\Lam,\lambda}$ to complex $\lambda$.  Where our definition \eqref{eq1ptDef} differs from Ruelle, however, is the use of a ``multivariate'' $\bl$ rather than constant $\lambda$ which allows us to write $\rho_{\Lam,\bl}$ as $\bl$ times a ratio of partition functions.  This turns out to be a crucial feature of \eqref{eq1ptDef} and essentially allow us to prove zero-freeness of $Z_\Lambda(\bl)$ inductively (see Theorem \ref{thmZeroFreeProof}).
	\end{remark*}

	In what follows in Sections~\ref{secDensities},~\ref{secComplexContraction},and~\ref{secMainProof} the region $\Lam$ will be fixed, and so we will drop the subscript $\Lam$ from the notation, writing $Z(\bl)$ for $Z_{\Lam}(\bl)$ and $\rho_{\bl}(x)$ for $\rho_{\Lam, \bl}(x)$.  We will interpret an activity function $\bl$ on $\Lam$ as a function $\bl : \R^d \to \mathbb C$ that is $0$ on $\R^d \setminus \Lam$.  In fact, Lemma~\ref{lemZIntegral} and Theorem~\ref{thmMainIdentitiy} below hold for bounded, integrable functions $\bl : \R^d \to \mathbb C$ without requiring bounded support.  We call such functions \textit{activity functions}.  
	
	Given an activity function $\bl$, let $\mathcal A_{\bl}  = \{\bl' = \alpha \cdot \bl : \alpha  : \R^d \to [0,1] \text{ is measurable} \}$.  
	We will use the following hereditary notion of zero freeness.  
	\begin{defn}
		An activity function $\bl$ is totally zero-free  if $Z(\bl') \ne 0$  for all $\bl' \in \cA_{\bl}$. 
	\end{defn}
	
	Our first lemma is an integral identity for the log partition function. 
	\begin{lemma}\label{lemZIntegral}
		If an activity function $\bl$ is totally zero-free then 
		\begin{align*}
		\log  Z(\bl) &= \int_{\R^d} \rho_{\hat{\bl}_x}(x)\,dx
		\end{align*}
		where
		$$ \hat{\bl}_x(y) = \begin{cases}
		0 & \text{if }y \in \Lambda_x \\
		\bl(y) &\text{if }y \notin \Lambda_x
		\end{cases}$$
		and $\Lambda_x = \{y \in \R^d : \dist(0,y) < \dist(0,x)  \}$.
	\end{lemma}
	\begin{proof}
		Define $\bl_t(y) = \one_{ \mathrm{dist}(y,0) \geq t } \bl(y)$ and note by assumption $Z(\bl_t) \neq 0$ for all $t$.  We will apply the fundamental theorem of calculus and integrate $\frac{d}{dt}\log Z(\bl_t)$.  Since we want to compute $\frac{d}{dt}\log Z(\bl_t)$, we compute $\frac{d}{dt} Z(\bl_t)$;  first change coordinates to write \begin{align*}
		Z(\bl_t) -1 &= \sum_{j \geq 1} \frac{1}{j!} \int_{(\R^d)^j} \prod_{i = 1}^j \bl_t(v_i) e^{-U(v_1,\ldots,v_j)}\,d\mathbf{v} \\
		&= \sum_{j \geq 1} \frac{1}{j!} \int_{0}^\infty \int_{\partial (B_s(0)^j)} \prod_{i = 1}^j \bl_t(v_i) e^{-U(v_1,\ldots,v_j)}\,d\mathbf{v} \,ds \\
		&= \sum_{j \geq 1} \frac{1}{j!} \int_{0}^\infty j \int_{\partial B_s(0)} \bl_t(w) \int_{B_s(0)^{j-1}} \prod_{i = 1}^{j-1} \bl_t(v_i) e^{-U(v_1,\ldots,v_{j-1},w)}\,d\mathbf{v}\,dw \,ds \\
		&= \sum_{j \geq 1} \frac{1}{j!} \int_{t}^\infty j \int_{\partial B_s(0)} \bl(w) \int_{B_s(0)^{j-1}} \prod_{i = 1}^{j-1} \bl_t(v_i) e^{-U(v_1,\ldots,v_{j-1},w)}\,d\mathbf{v}\,dw \,ds
		\end{align*}
		and so we have
		\begin{align*}
		\frac{d}{dt}Z(\bl_t) &= -\sum_{j \geq 1} \frac{1}{j!}  \cdot j \int_{\partial B_t(0)} \bl(w) \int_{(\R^d)^{j-1}} \prod_{i = 1}^{j-1} \bl_t(v_i) e^{-U(v_1,\ldots,v_{j-1},w)}\,d\mathbf{v}\,dw
		\end{align*}
		where the outermost integral is a $d-1$ dimensional integral over the boundary of the ball of radius $t$ centered at $0$.  Rearranging the terms, we have  
		\begin{align*}
		\frac{d}{dt}Z(\bl_t) &= - \int_{\partial B_t(0)} \bl(w) \sum_{j \geq 0} \frac{1}{j!} \int_{(\R^d)^{j}} \prod_{i = 1}^{j} \left(\bl_t(v_i) e^{-\phi(v_i - w)}\right) e^{-U(v_1,\ldots,v_{j})}\,d\mathbf{v}\,dw \\
		&= - \int_{\partial B_t(0)}  \bl(w) Z(\bl_t e^{-\phi(w - \cdot)}) \,dw\,.
		\end{align*}
		The fundamental theorem of calculus then gives \begin{align*}
		\log Z(\bl_\infty) - \log Z(\bl_0) &= \int_0^\infty \frac{d}{dt} \log(Z(\bl_t))\,dt \\
		&= - \int_0^\infty\int_{\partial B_t(0)} \bl(w) \frac{Z(\bl_t e^{-\phi(w - \cdot)})}{Z(\bl_t)} \,dw \,dt \\
		&= -\int_{\R^d} \rho_{\hat{\bl}_w}(w)\,dw\,.
		\end{align*}
		Noting $Z(\bl_\infty) = 1$ and $Z(\bl_0) = Z(\bl)$ completes the Lemma.
	\end{proof}

	Next we give a recursive identity for the densities.  Fix an activity function $\bl$.  For $v, w \in \R^d$ let $\bl_{v \to w}: \R^d \to \mathbb C$ be defined by
	\[ \bl_{v \to w}(x) = \begin{cases} \bl (x) e^{-\phi(v-x)} &\mbox{if } \dist(v,x) < \dist(v,w) \\ 
	\bl(x) &\mbox{if } \dist(v,x) \ge \dist(v,w)  \,.\end{cases} \]    
	In particular, $\bl_{v \to w} \in \cA_{\bl}$.

	\begin{theorem}
		\label{thmMainIdentitiy}
		Suppose an activity function $\bl$ is totally zero-free.  Then for every $v \in \R^d$ we have 
		\begin{equation}
		\label{eqIdentity}
		\rho_{\bl}(v) = \bl (v) \exp\left(-\int_{\R^d}\rho_{\bl_{v \to w}}(w)(1 - e^{-\phi(v-w)})     \,dw  \right)\,.
		\end{equation}
	\end{theorem}
	\begin{proof}
		For $v$ fixed and each $t \geq 0$ define $$\bl_t(x) = \begin{cases}
		\bl(x)e^{-\phi(v-x)} &\text{ if } \dist(v,x) < t \\
		\bl(x) &\text{ otherwise}
		\end{cases}\,.$$
		Note that $\bl_t  \in \cA_{\bl}$, and so by assumption $Z(\bl_t) \ne 0$. 
		
		It will be helpful to write $\bl_t(x) = \bl(x)(1 - \one_{\dist(v,x) \leq t}(1 - e^{-\phi(v-x)}))$.  As in the proof of Lemma \ref{lemZIntegral}, we will apply the fundamental theorem of calculus to $\frac{d}{dt} \log Z(\bl_t)$. We start by computing $\frac{d}{dt} Z(\bl_t)$.  Write \begin{align*}
		Z(\bl_t) &= \sum_{k \geq 0} \frac{1}{k!} \int_{(\R^d)^k} \prod_{j = 1}^k\left(\bl(x_j)(1 - \one_{\dist(v,x_j) \leq t}(1 - e^{-\phi(v-x_j)})\right) \exp\left(-U(x_1,\ldots,x_k)\right)\,d\mathbf{x}\,.
		\end{align*}
		
		Note that for each term in the product, we have $$\frac{d}{dt} (1 - \one_{\dist(v,x_j) \leq t}(1 - e^{-\phi(v-x_j)})) = -\delta_{\dist(v,x_j)}(t)(1 - e^{-\phi(v-x_j)})$$
		where $\delta_{\dist(v,x_j)}(t)$ is the Dirac delta function; equivalently, change coordinates to write $Z(\bl_t)$ as in the proof of Lemma \ref{lemZIntegral} and use the fundamental theorem of calculus.  By the product rule, this implies \begin{align*}
		\frac{d}{dt}Z(\bl_t) &= -\sum_{k \geq 0} \frac{1}{k!} \int_{(\R^d)^{k-1}} \int_{\partial B_t(v)} \prod_{j = 1}^{k-1}\bl_t(x_j) e^{-U(\mathbf{x})}(1 - e^{-\phi(w-v)})\cdot k \bl(w) e^{-\sum_{j = 1}^{k-1} \phi(x_j - w)}\,d\mathbf{x}\,dw \\
		&= - \int_{\partial B_t(v)} \bl(w)  Z(\bl_t e^{-\phi(w - \cdot)}) (1 - e^{-\phi(w-v)})\,dw\,.
		\end{align*}
		
		By the fundamental theorem of calculus, we then have \begin{align*}
		\log(Z(\bl_{\infty})) - \log(Z(\bl_0)) &= \int_0^\infty \frac{\frac{d}{dt} Z(\bl_t)}{Z(\bl_t)} \,dt \\
		&= -\int_0^\infty \int_{\partial B_t(v)} \bl(w) \frac{Z(\bl_t e^{-\phi(w-\cdot)})}{Z(\bl_t)} (1 - e^{-\phi(w-v)})\,dw\,dt \\
		&= - \int_{\R^d} \rho_{\bl_{v \to w}}(w)(1 - e^{-\phi(v-w)})\,dw\,.
		\end{align*}
		Noting $\bl_\infty = \bl e^{-\phi(v - \cdot)}$ and $\bl_0 = \bl$ and applying \eqref{eq1ptDef} completes the proof.
	\end{proof}

	Finally we need two continuity statements.
	\begin{lemma}
		\label{lemContinuous}
		For any region $\Lam$ and any $M>0$, the map  $\bl \mapsto Z_\Lam(\bl)$ is uniformly continuous in the sup norm on the set of activity functions $\bl$ with $|\bl(x)| \le M$ for all $x \in \Lam$. 
	\end{lemma}
	To prove this we need the following elementary lemma that appears in~\cite{durrett2019probability}; we reproduce the simple proof for completeness: 
	\begin{lemma}\label{lem:prod-diff}
		Let $\{z_j\}_{j = 1}^n$ and $\{w_j\}_{j = 1}^n$ be complex numbers with $|z_j|, |w_j| \leq \theta$ for all $j$.  Then $$\left|\prod_{j = 1}^n z_j - \prod_{j = 1}^n w_j \right| \leq \theta^{n-1}\sum_{j = 1}^n |w_j - z_j|\,.$$
	\end{lemma}
	\begin{proof}
		Calculate \begin{align*}
		\left|\prod_{j = 1}^n z_j - \prod_{j = 1}^n w_j \right| &\leq \left|z_n\prod_{j = 1}^{n-1}z_j - z_n\prod_{j = 1}^{n-1} w_j \right| + \left|z_n \prod_{j = 1}^{n-1} w_j - w_n \prod_{j = 1}^{n-1} w_j \right| \\
		&\leq \theta \left| \prod_{j = 1}^{n-1} z_j - \prod_{j = 1}^{n-1} w_j\right| + \theta^{n-1} |z_n - w_n|
		\end{align*}
		and induct on $n$.
	\end{proof}
	Now we prove Lemma~\ref{lemContinuous}. 
	\begin{proof}[Proof of Lemma \ref{lemContinuous}]
		Suppose $\bl, \bl'$ are bounded by $M \ge 1$ in absolute value and $| \bl(x) - \bl'(x)| \leq \del$ for all $x \in \Lam$.  Then
		\begin{align*}
		| Z_{\Lam}(\bl ) - Z_{\Lam}(\bl')| &\le \sum_{k \ge 1} \frac{1}{k!} \int_{\Lam^k}   e^{- U(x_1, \dots, x_k) }  \cdot \left |   \prod_{i=1}^k \bl(x_i)   -  \prod_{i=1}^k \bl'(x_i)   \right |d x_1 \cdots dx_k    \\
		&\le \sum_{k \ge 1} \frac{1}{k!} \int_{\Lam^k}      \left |   \prod_{i=1}^k \bl(x_i)   -  \prod_{i=1}^k \bl'(x_i)   \right |d x_1 \cdots dx_k   \\
		&\le  \sum_{k \ge 1} \frac{1}{k!}  |\Lam|^k  M^k  k \del \\
		&\le \del |\Lam| M e^{|\Lam| M} \,,
		\end{align*}
		and so $ Z_{\Lam}(\cdot )$ is uniformly continuous. 
	\end{proof}
	
	\begin{lemma}
		\label{lemRhoContinuous}
		For any region $\Lam$, if an activity function $\bl$ on $\Lam$ is bounded and totally zero-free then for every $v \in \Lam$, $\rho_{\bl'}(v)$ is uniformly continuous in $\bl'$ in a neighborhood of $\bl$ where the modulus of continuity is uniformly bounded as $v$ varies.
	\end{lemma}
	\begin{proof}
		Since $\bl$ is totally zero-free, by the uniform continuity of $Z$, $\bl'$ is totally zero-free and $\rho_{\bl'}(v)$ is well defined for $\bl'$ in a neighborhood  of $\bl$ in the sup norm.  Since $Z_\Lambda(\bl) \neq 0$, by Lemma \ref{lemContinuous} there is some $\eps > 0$ so that  $|Z_\Lam(\bl')|$ is uniformly bounded away from $0$ for $\bl'$ so that $\| \bl - \bl' \|_\infty < \eps$.  Another application of Lemma \ref{lemContinuous} shows that both $\frac{1}{Z_\Lambda(\bl')}$ and $Z_\Lam(\bl' e^{- \phi(v - \cdot)})$  are uniformly continuous for $\| \bl - \bl'\|_\infty < \eps$ and so $\rho_{\bl'}(v) = \bl'(v)\frac{ Z_\Lam(\bl' e^{- \phi(v - \cdot)})   }{  Z_\Lam(\bl')}$ is uniformly continuous in this neighborhood.
	\end{proof}
	
	\section{Complex contraction}
	\label{secComplexContraction}

	For $\lam \in \mathbb C$, and  $ \brho : \R^d \to \mathbb C$  bounded and integrable, define 
	\begin{align*}
	F(\lam, \brho) &=  \lam  \exp\left(-\int_{\R^d} \brho(w)(1 - e^{-\phi(w)})     \,dw  \right) \,.
	\end{align*}
	This function arises on the right hand side of the identity in~\eqref{eqIdentity}.  We will show that it is contractive in an appropriate sense on a complex domain.  This contractive property is similar to that shown for a different function in the case of discrete two-spins systems in~\cite{peters2019conjecture,shao2019contraction}. 
	
	For a positive reals $s, \eps$ we let $\mathcal N(s,\eps) =  \{ y \in \mathbb C: \exists x \in [0, s],  | y - x | < \eps \}$ and $\overline{\mathcal N}(s,\eps) =  \{ y \in \mathbb C: \exists x \in [0, s],  | y - x | \le \eps \}$.
	\begin{lemma}
		\label{lemContract}
		For every $\lam_0 \in ( 0 , e/ C_{\phi})$ there exists $ \eps >0$ and complex neighborhoods $U_1 \subset U_2$ so that  $[0,e/C_\phi] \subset U_1$ with $\overline{U}_1 \subset U_2$ so that the following holds.  If $\lam\in   \mathcal N(\lam_0, \eps)$ and  $\brho (x) \in \overline{U}_2$ for all $x \in \R^d$, then $F(\lam, \brho) \in U_1$. 
	\end{lemma}
	
	We will prove Lemma~\ref{lemContract} in a sequence of steps.  We will apply a change of coordinates in order to make $F$ a contraction mapping, as was done in the discrete setting in, e.g.~\cite{restrepo2013improved,li2013correlation,peters2019conjecture,shao2019contraction}.  With this in mind, set $\psi(x) := \log(1 + C_\phi x)$.  For a function $\bz:\R^d \to [0,\log(1+e)]$, define $g_\lambda(\bz) := \psi(F(\lambda,\psi^{-1}(\bz) ) )$.  First, we will consider only constant functions $\bz \equiv z$; our first main step is to show that $|g'_\lambda(z)| \leq 1$ for all $z$ and all $\lambda \in [0,e/C_{\phi}]$.   Thus, by definition $$g_\lambda (z) = \log(1 + C_\phi \lambda e \cdot e^{-e^z})$$
	and so $$g_{\lam}'(z) = - \frac{C_\phi \lambda e \cdot e^{z} \cdot e^{-e^z}}{1 + C_\phi \lambda e \cdot e^{-e^z}}\,.$$
	
	\begin{lemma}
		For all $z \geq 0$ we have $|g_{e/C_\phi}'(z)| \leq 1$. 
	\end{lemma}
	\begin{proof}
		Set $r = e^{z} \in [1,\infty)$ and note that then we may write $$h(r):=|g_{e/C_{\phi}}'(z)| = \frac{e^2 r e^{-r}}{1 + e^2 e^{-r}}\,.$$
		Compute $$h'(r)= \frac{e^2 e^{-r}(e^{2-r} -r + 1)}{(1 + e^{2-r})^2}$$
		which has a unique zero at $r = 2$, showing that $h(r) \leq h(2) = 1$.  Noting $h(r) \geq 0$ completes the lemma.
	\end{proof}
	
	Now we can show that for $\lambda < e/C_\phi$ we have that $|g_\lam'|$ is \emph{strictly} less than $1$, implying that $g_\lam$ is a contraction.
	\begin{lemma}\label{lem:g-diff}
		For each $\lambda_0 \in [0,e/C_\phi)$ there exists $\delta > 0$ so that for all $\lambda \in [0,\lambda_0]$, we have $$\max_{z \geq 0} |g_\lambda'(z)| \leq 1 - \delta\,.$$
	\end{lemma}
	\begin{proof}
		Set $\eps = e/C_{\phi} - \lam_0 >0$.  Set $\beta = C_\phi e \lambda$ and note that $0 \leq \beta \leq C_\phi e ( e/C_\phi - \eps) = e^2(1 - C_\phi \eps / e)$.  Assume $\beta > 0$, as otherwise $g_\lambda' = 0$.  Write \begin{align*}
		\max_{z  \geq 0} \left|\frac{g_\lambda'(z)}{g_{e/C_\phi}(z)} \right| &= \max_{z \geq 0} \frac{e^{-2} + e^{-e^z} }{\beta^{-1} + e^{-e^z}} \\
		&= \max_{z \geq 0} \left(1 - \frac{\beta^{-1} - e^{-2}}{\beta^{-1} + e^{-e^z}}  \right) \\
		&= 1 - \frac{\beta^{-1} - e^{-2}}{\beta^{-1} + e}\,.
		\end{align*}
		Choosing $\delta:= \min_{\beta \in [0,(1-\eps C_\phi/e)e^2]} \frac{1 - \beta e^{-2}}{1 + \beta e}$ completes the proof.
	\end{proof}
	
	We now extend the definition of $g_{\lam}(\bz)$ to take complex values of $\lam$ and complex-valued functions $\bz$.  Our next lemma is for the special case of constant, real $\bz$. 
	\begin{lemma}\label{lem:gl-diff}
		Fix $\lambda_0 \in [0,e/C_\phi)$.  Then there exists $M > 0$ and $\eps_1 > 0$ so that
		$$\max_{\lambda \in \mathcal{N}(\lambda_0,\eps_1),z \in \R} \left| \frac{d}{d\lambda} g_\lambda(z)\right| \leq M \,.$$
	\end{lemma}
	\begin{proof}
		This follows from computing $$\frac{d}{d\lambda} g_\lambda(z) = \frac{C_\phi e \cdot e^{-e^z}}{1 + C_\phi \lambda e \cdot e^{-e^z}}\,.$$
	\end{proof}
	
	We now prove a version of Lemma \ref{lemContract} for the case of $\bz \equiv z$: 
	\begin{lemma}\label{lem:one-pt-contraction}
		Fix $\lambda_0 \in [0,e/C_\phi)$.  There exists $\eps_1, \eps_2 >0$ and $\eps_3 \in (0, \eps_2)$ so that for any $\lambda \in \mathcal{N}(\lambda_0,\eps_1)$ and $z \in \overline{\mathcal N}(\log(1 + e),\eps_2)$ we have $g_\lambda(z) \in \mathcal{N}(\log(1 + e),\eps_3)$. 
	\end{lemma}
	\begin{proof}
		First note from Lemma \ref{lem:g-diff}, there is a $\delta > 0$ so that $|g_{\lambda}'(z)| \leq 1 - \delta$ for all $\lambda \in [0,\lambda_0]$ and $z \in [0,\log(1 + e)]$.  Since $g_\lambda'$ is continuous in $z$ and $\lambda$, we may take $\eps_1,\eps_2$ sufficiently small so that \begin{equation}\label{eq:gprime}
		|g_\lambda'(z)| \leq 1 - \delta/2, \quad \text{for all } (z,\lambda) \in \overline{\mathcal N}(\log(1 + e),\eps_2) \times \mathcal{N}(\lambda_0,\eps_1)\,.
		\end{equation}
		To see this, note that for each $\lambda \in [0,\lambda_0]$ and $z \in [0,\log(1+e)]$ we may find a neighborhood of $(\lambda,z)$ for which the above holds, thereby giving an open cover of $[0,\lambda_0] \times [0,\log(1+e)]$; by compactness, we may reduce to a finite cover, thereby giving $\eps_1, \eps_2$.

		By Lemma \ref{lem:gl-diff}, we may take $\eps_1$ small enough so that \begin{equation}\label{eq:glam}
		\max_{\lambda \in \mathcal{N}(\lambda_0,\eps_1),z \in [-\eps_2,\log(1 + e)]} \left| \frac{d}{d\lambda} g(z)\right| \leq \frac{\delta \eps_2}{4 \eps_1}\,.
		\end{equation}
		
		For $\lambda \in  \mathcal{N}(\lambda_0,\eps_1)$ and $z \in \overline{\mathcal N}(\log(1 + e),\eps_2)$, find $z' \in [0,\log(1+e)]$ with $|z - z'| \leq \eps_2$ and $\lambda' \in [0,\lambda_0]$ with $|\lambda - \lambda'| \leq \eps_1$.  Using the bounds \eqref{eq:gprime} and \eqref{eq:glam}, bound 
		\begin{align*}
		|g_\lambda(z) - g_{\lambda'}(z')| &\leq |g_\lambda(z) - g_\lambda(z')| + |g_{\lambda}(z') - g_{\lambda'}(z')| \\
		&< (1 - \delta/2)\eps_2 + \eps_2\delta/ 4 \\
		&= \eps_2(1 - \delta/4) \\
		&=: \eps_3\,.
		\end{align*}
		Since $g_\lambda([0,\log(1 + e)]) \subset [0,\log(1 +e)]$ we have $g_\lambda(z) \in [0,\log(1+e)]$, thereby completing the proof.
	\end{proof}

	To generalize Lemma \ref{lem:one-pt-contraction} to non-constant $\bz$, we will use a convexity argument.  We need the following fact that appears in~\cite[Proof of Lemma $4.1$]{peters2019conjecture}.
	\begin{fact}\label{fact:convex}
		For $r > 0$ and $\eps$ sufficiently small, the image of $\overline{\mathcal{N}}(r,\eps)$ under the map $z \mapsto e^z$ is convex.
	\end{fact}
	
	This allows us to generalize Lemma \ref{lem:one-pt-contraction} almost immediately.
	\begin{lemma}\label{lem:cc-potential-world}
		Fix $\lambda_0 \in [0,e/C_\phi)$.  There exists $\eps_1, \eps_2>0$ and $\eps_3 \in (0,\eps_2)$ so that for any $\lambda \in \mathcal{N}(\lambda_0,\eps_1)$ and any $\bz$ so that $\bz(x) \in \overline{\mathcal N}(\log(1 + e),\eps_2)$ for all $x$, we have $g_\lambda(\bz) \in \mathcal{N}(\log(1 + e),\eps_3)$. 
	\end{lemma}
	\begin{proof}
		We will show that there is some $z \in \overline{\mathcal N}(\log(1 + e),\eps_2)$ so that $g_{\lam}(\bz) = g_{\lam}(z)$.  This would follow if we find such a $z$ so that
		$$C_\phi e^z = \int_{\R^d} e^{\bz(w)} (1 - e^{-\phi(w)})\,dw\,.$$
		Set $\alpha_w := (1 - e^{-\phi(w)})/C_\phi$, and note that $\alpha_w \geq 0$ with $\int \alpha_w = 1$; in particular, $\alpha_w \,dw$ is a probability measure.  Fact \ref{fact:convex} implies there is some such $z$ so that $e^z = \int e^{\bz(w)} \alpha_w \,dw$.  Then we can apply Lemma~\ref{lem:one-pt-contraction} to complete the proof.
	\end{proof}

	\begin{proof}[Proof of Lemma \ref{lemContract}]
		Find $\eps_1, \eps_2$ and $\eps_3$ according to Lemma \ref{lem:cc-potential-world}.  Then we claim we may take $\eps = \eps_1$ and $U_2 = \psi^{-1}(\mathcal{N}(\log(1+e),\eps_2)), U_1 = \psi^{-1}(\mathcal{N}(\log(1+e),\eps_3))$.  The only thing to check is that $U_j$ are open and bounded, and the boundaries of $\mathcal{N}(\log(1+e),\eps_j)$ are mapped to the boundaries of their inverse image under $\psi$; since $\psi^{-1}(z) = (e^z - 1)/C_\phi$ is analytic with non-vanishing derivative, it is a conformal map and thus these properties hold.
	\end{proof}

	\section{Zero freeness }
	\label{secMainProof}
	
	In this section we use the tools from the previous sections to prove Theorem~\ref{thmZeros}.  It is contained in the following stronger theorem.

	\begin{theorem}\label{thmZeroFreeProof}
		For every $\lam_0 \in ( 0 , e/ C_{\phi})$ there exists $\eps >0$ and $C>0$ so that the following holds for every region $\Lam \subset \R^d$.  Let $\bl$ be an activity function on $\Lam$  and suppose $\bl(x) \in \mathcal N ( \lam_0, \eps)$ for all $x \in \Lam$.  Then
		\begin{align*}
		| \log  Z_\Lam(\bl) | &\le C| \Lam| \,.
		\end{align*}
		In particular, $Z_\Lam(\bl)  \ne 0$.   Furthermore, for every $v \in \Lam$, $|\rho_{\bl}(v)| \leq C$.
	\end{theorem}
	\begin{proof}
		Fix $\Lam$.  Fix $\lam_0$ and let $\eps >0$ and $U_1, U_2$ be as guaranteed in Lemma~\ref{lemContract}.  Let $C = \max \{ |z| : z\in \overline U_2 \}$.  Fix some $\bl$ so that $\bl(x) \in \mathcal N ( \lam_0, \eps)$ for all $x$.  We will show that for all $\bl ' \in \cA_{\bl}$,
		\begin{equation}
		\label{eqRhoBound}
		\rho_{\bl'}(v) \in  \overline{U}_2  \text{ for all } v \in \Lam \,. 
		\end{equation}

		Towards this end let $\cA_* \subseteq \cA_{\bl}$ be the subset of activity functions $\bl'$ for which~\eqref{eqRhoBound} holds for all $\bl '' \in \cA_{ \bl'}$.  In particular, if $\bl' \in \cA_*$, then $\bl'$ is totally zero-free.  
		
		First we observe that if $\bl' \in \cA_*$, then since  $\hat \bl'_{x} \in \cA_{\bl'}$ (where $\hat \bl'_{x}$ is as in Lemma~\ref{lemZIntegral}), we have  $|\rho_{\hat \bl'_{x}} (x) |  \le C$.  Applying Lemma~\ref{lemZIntegral} then gives that  
		\begin{equation}
		\label{eqlogZCbound}
		| \log Z(\bl') | \le \int_{\Lam}   |\rho_{\hat \bl'_{x}} (x) | \, dx \le C |\Lam| \,.
		\end{equation}
		
		Our goal is now to show that $\cA_* = \cA_{\bl}$, or equivalently that $\bl \in \cA_*$.   The identically $0$ activity function is in $\cA_*$ since $\rho_{0}(v) =0$.  Now suppose for the sake of contradiction that $\cA_* \ne \cA_{\bl}$.  Let $\bl_1$ be an arbitrary activity function in $\cA_{\bl} \setminus \cA_*$ and for $t \in[0,1]$ let $\bl_t = t \bl_1$.  We have $0 = \bl_0 \in \cA_*$, $\bl_1 \notin \cA_*$,  and $\bl_t \in \cA_* \Rightarrow \bl_{t'} \in \cA_*$ if $t' < t$.  We can then define $t^* = \sup \{ t: \bl_t \in \cA_* \}$ and set $\bl_* = \bl_{t^*}$.   By the uniform continuity of $Z( \cdot)$ and $\rho_{\cdot}(x)$ around $\bl \equiv 0$ (Lemmas~\ref{lemContinuous} and~\ref{lemRhoContinuous}), we know that $t^* >0$. 
		
		Now for any $t < t^*$,  $\bl_t \in \cA^*$ and so by~\eqref{eqlogZCbound}, $|\log Z(\bl_t) | \le C |\Lam|$.  Since this is true for all $t < t^*$, by the uniform continuity of $Z( \cdot)$ we have that $\bl_\ast$ is totally zero-free and thus $\rho_\cdot$ is uniformly continuous in a neighborhood of $\bl_*$.  This then shows $\bl_* \in \cA_*$.

		Moreover, by uniform continuity again, we have that $\bl'$ is totally zero-free for all $\bl'$ in a neighborhood of $\bl_*$ in the sup norm and thus $\rho_{\bl'}(v)$ is well defined for all $x$ and all $\bl'$ in this neighborhood. 
		
		Now consider $\rho_{\bl_*}(x)$ for an arbitrary $x \in \Lam$.  We know that for each $w$, $\rho_{(\bl_*)_{v\to w}}(w) \in   \overline{U}_2$ since $(\bl_*)_{v\to w} \in \cA_{ \bl_*}$ and thus $\rho_{(\bl_*){v\to w}} \in \cA_*$.  Then applying Theorem~\ref{thmMainIdentitiy} and Lemma~\ref{lemContract} we see that $\rho_{\bl_*}(v) \in U_1$.   But then by the uniform continuity of $\rho_{\bl_*}(x)$ in a neighborhood of $\bl_*$ and the fact that $\dist(\overline{U}_1, U_2^c) > 0$,  we see that  $\bl' \in \cA_*$ for all $\bl' $ in a neighborhood of $\bl_*$.    In particular, $\bl_{t'} \in \cA_*$ for some $t' > t_*$, a contradiction.  Thus we have that $\cA_* = \cA_{\bl}$ as desired.

		Now since $\bl \in \cA_*$, applying~\eqref{eqlogZCbound} again gives $|\log Z(\bl)| \le C |\Lam|$ as desired.
	\end{proof}

	\section{Proof of Theorem~\ref{thmAnalytic}, Theorem~\ref{corDensity}, and Corollary~\ref{corHSdensity}}
	\label{secFinalProofs}

	To prove analyticity of the pressure and uniqueness of Gibbs measure, we will  use Vitali's Theorem (see, for instance, \cite[Theorem 6.2.8]{simon2015basic}).
	\begin{theorem}[Vitali's Convergence Theorem]\label{thmVitali}
		Let $\Omega \subset \C$ be a complex neighborhood and $f_n$ a sequence of analytic functions on $\Omega$ so that $|f_n(z)| \leq M$ for some $M$ and for all $n$ and all $z \in \Omega$.  If there is a sequence of distinct numbers $z_m \in \Omega$ with $z_m \to z_\infty \in \Omega$ so that $\lim_{n \to \infty} f_n(z_m)$ exists for all $m$, then $f_n$ converges uniformly on compact subsets of $\Omega$ to an analytic function $f$ on $\Omega$.
	\end{theorem}

	We will also need some basic properties of boundary conditions, which we will define via locally finite points sets.  
	
	We say a locally finite point set $\bY \subset \R^d$  is \emph{tempered} with respect to a potential $\phi$ if for almost all $v \in \R^d$ we have
	$$ \lim_{n \to \infty} \sum_{y \in \bY \cap \Lambda_n^c}  \phi(v- y) = 0 \,,$$
	 where we recall that $\Lambda_n$ is the ball of volume $n$ centered at $0$.   We will need the Georgii, Nguyen, Zessin (GNZ) equations for Gibbs measures (see, e.g., \cite{dereudre2019introduction} or \cite{jansen2019cluster}): if $\bY$ is a sample of a Gibbs measure $\nu$ associated to activity $\bl$ then for any positive measurable function $F:\R^d \to \R$ we have  
	 \begin{equation}	\label{eq:GNZ}
	 	\E \sum_{y \in \bY} F(y) = \int_{\R^d} \bl(x) F(x)\E e^{-\sum_{y \in \bY}\phi(y - x) } \,dx\,.
	 \end{equation}
	 
	 We first show that almost every sample from a Gibbs measure is tempered.
	\begin{lemma}\label{lem:bc-tempered}
		Let $\nu$ be a Gibbs measure on $\R^d$ associated to an activity $\bl$ bounded by $\lambda < \infty$ and a tempered, repulsive potential $\phi$.  Then $\nu$-almost-every  $\bY$ is tempered.
	\end{lemma}
	\begin{proof}
		By the GNZ equations \eqref{eq:GNZ} note that for each $v$ we have $$\E |1 - e^{-\sum_{x \in \bY \cap \Lambda_n^c}\phi(v - x)}| \leq \E \sum_{x \in \bY \cap \Lambda_n^c} (1 - e^{-\phi(v - x)}) \leq \lambda \int_{\Lambda_n^c} (1 - e^{-\phi(v - x)})\,dx.$$
		Thus, for a bounded set $B$ we have $$\lim_{n\to \infty}\E \int_{B}|1 - e^{-\sum_{x \in \bY \cap \Lambda_n^c}\phi(v - x)}|\,dv  \leq \lambda |B| \lim_{n \to \infty} \int_{\Lambda_n^c} (1 - e^{-\phi(v - x)})\,dx = 0 $$
		by temperedness of $\phi$.  Since the random variables $\int_{B} |1 - e^{-\sum_{x \in \bY \cap \Lambda_n^c} \phi(v - x)}| \,dv$ are monotone decreasing in $n$, we have $\int_{B} |1 - e^{-\sum_{x \in \bY \cap \Lambda_n^c} \phi(v - x)}| \,dv \xrightarrow{n\to\infty} 0$ for almost every $\bY$ for each fixed $B$.  Applying this for $B = \Lambda_m$ for each integer $m$ shows that for almost every $\bY$ we have $1 -e^{-\sum_{x \in \bY \cap \Lambda_n^c} \phi(v - x)} \to 0$ for almost-all $v$; this shows that $\bY$ is tempered.
	\end{proof}

	 We now define boundary conditions in terms of locally finite points sets $\bY$. Given a bounded measurable set $\Lambda$, the Gibbs measure associated to  activity $\bl$ on $\Lambda$ with boundary condition $\bY$  is  the Gibbs  point process on $\Lambda$ with activity $\bl_{\bY,\Lambda}$ defined by $\bl_{\bY,\Lambda}(v) = \bl(v) \exp\left(-\sum_{y \in \bY \cap \Lambda^c} \phi(v-y)  \right)$.  Given a tuple $(v_1,\ldots,v_k) \in (\R^d)^k$, a bounded set $\Lambda$, boundary condition $\bY$, and activity $\bl$, we will let $\rho_{\Lambda,\bl}^{\bY}(v_1,\ldots,v_k) := \rho_{\Lambda, \bl_{\bY,\Lambda}}(v_1,\ldots,v_k)$ denote the $k$-point density under boundary condition $\bY$; define the partition function $Z_\Lambda^{\bY}(\bl)$ analogously.  	
	We next show that uniqueness of infinite-volume Gibbs measure at an activity $\lambda$ follows from uniqueness of one-point densities for all activities below $\lambda$ under tempered boundary conditions.

	\begin{lemma}\label{lem:bc-uniqueness}
		Fix a tempered, repulsive potential $\phi$ and $\lambda > 0$.  Let $\Lambda_n$ denote the ball of volume $n$ centered at $0$ in $\R^d$.   Suppose that for all $\bl$ bounded by $\lambda$ and all tempered boundary conditions $\bY$ and $\bY'$  we have $\lim_{n \to \infty} \rho_{\Lambda_n,\bl}^{\bY}(v) = \lim_{n \to \infty} \rho_{\Lambda_n,\bl}^{\bY'}(v)$.  Then there is a unique infinite volume Gibbs measure with potential $\phi$ at activity $\lambda$. 
	\end{lemma}
	\begin{proof}
		For an infinite volume Gibbs measure $\mu$ associated to a tempered, repulsive potential $\phi$ at activity $\lambda$, let $\rho_\mu(v_1,\ldots,v_k)$ denote the $k$-point correlation function associated to it.  Since $\phi$ is repulsive, we have the Ruelle bound $\rho_\mu(v_1,\ldots,v_k) \leq \lambda^k$, and so $\mu$ is determined by its collection of $k$-point density functions (see \cite{ruelle1999statistical, jansen2019cluster}).  Thus, it is sufficient to show that $\rho_{\mu_1}(v_1,\ldots,v_k)=\rho_{\mu_2}(v_1,\ldots,v_k)$  for each $k \geq 1$, all $(v_1,\ldots,v_k)$, and each pair of  infinite volume Gibbs measures $\mu_1 ,\mu_2$. 
		
		By the DLR equations (see \cite{ruelle1999statistical,jansen2019cluster}) we have $$
		\rho_\mu(v_1,\ldots,v_k) = \E_{\mu} \rho_{\Lambda_n,\lambda}^{\bY}(v_1,\ldots,v_k)
		$$
		where the expectation is over all samples of $\bY$ from $\mu$.  Since  $\bY$ is tempered with probability $1$, we may take the expectation only over tempered $\bY$.  For a pair of Gibbs measures $\mu_1$ and $\mu_2$ we may apply Fatou's lemma to find tempered boundary conditions $\bY$ and $\bY'$ so that $$\rho_{\mu_1}(v_1,\ldots,v_k) \leq \limsup_{n \to\infty}\rho_{\Lambda_n,\lambda}^{\bY}(v_1,\ldots,v_k), \quad  \liminf_{n\to\infty} \rho_{\Lambda_n,\lambda}^{\bY'}(v_1,\ldots,v_k) \leq \rho_{\mu_2}(v_1,\ldots,v_k)\,.$$
		Note that for the first bound we apply Fatou's Lemma to the sequence of non-negative random variables $\lambda^k - \rho_{\Lambda_n,\lambda}^{\bY}(v_1,\ldots,v_k)$. 		
		 Now write 
		\begin{align*}
		\rho_{\Lam_n, \lam}^{\bY} (v_1, \dots, v_k)  &=\lam^k e^{-U(v_1,\ldots,v_k)}  \frac{Z_{\Lam_n}^{\bY}(\lam e^{-  \sum_{i=1}^k\phi(v_i - \cdot)})}{Z_{\Lam_n}^{\bY}(\lam)}  \\
		&= \lam^k e^{-U(v_1,\ldots,v_k)}  \prod_{j=1}^k  \frac{ Z_{\Lam_n}^{\bY} (\bl_{j+1})}{ Z_{\Lam_n}^{\bY} (\bl_{j})} \\
		&= \prod_{j=1}^k \rho_{\Lam_n, \bl_j}^{\bY}(v_j)
		\end{align*}
		where $\bl_j = \lam e^{-  \sum_{i=1}^{j-1}\phi(v_i - \cdot)}$ (in particular $\bl_1 \equiv \lam$).  Perform the same decomposition for $\bY'$ and note that $\bl_j \leq \lambda$ for each $j$.  The assumption of the lemma then shows that $$\limsup_{n \to \infty} \rho_{\Lambda_n,\lambda}^{\bY}(v_1,\ldots,v_k) = \liminf_{n \to \infty} \rho_{\Lambda_n,\lambda}^{\bY'}(v_1,\ldots,v_k)$$
		and so $\rho_{\mu_1}(v_1,\ldots,v_k) \leq \rho_{\mu_2}(v_1,\ldots,v_k)$. Swapping roles of $\mu_1$ and $\mu_2$ allows us to conclude that $\rho_{\mu_1}(v_1,\ldots,v_k) = \rho_{\mu_2}(v_1,\ldots,v_k)$, thus completing the proof.		
	\end{proof}

We will shortly provide a route to deduce uniqueness of Gibbs measure at a given activity $\lam$ from two ingredients: (1) analyticity of densities for all activity functions $\bl$ bounded by $\lam$ and (2) uniqueness of Gibbs measure for all $\bl$ bounded by some (small) $\lam_0$.  These ingredients will allow us to verify the conditions of Lemma~\ref{lem:bc-uniqueness}.  The first ingredient is provided by  Theorem \ref{thmZeroFreeProof}; the second (with $\lam_0 = 1/(e C_{\phi})$) follows from  convergence of the cluster expansion, for example.  We deduce from this the following lemma.

	\begin{lemma}\label{lem:small-lam-uniqueness}
		Let $\lambda \in [0,1/(eC_\phi))$.  Then for any two tempered boundary conditions $\bY$ and $\bY'$ and any $\bl \leq \lambda$ we have 
		$$\lim_{n \to \infty} \rho_{\Lambda_n,\blambda}^{\bY}(v) = \lim_{n \to \infty} \rho_{\Lambda_n,\blambda}^{\bY'}(v)\,.$$
	\end{lemma}
	\begin{proof}
		Let $\mu_n$ denote the Gibbs measure on $\Lambda_n$ with activity $\bl$ and boundary condition $\bY$; similarly, let $\mu_n'$ denote the Gibbs measure on $\Lambda_n$ with activity $\bl$ and boundary condition $\bY'$.  Then by \eqref{eq1ptDef} we have $\rho_{\Lambda_n,\blambda}^{\bY}(v) = \bl(v) \E_{\mu_n} e^{-H_v(\bX)}$ where $H_v(\bX) = \sum_{x \in \bX} \phi(v - x)$.  For $n > R > 0$ define $H_{v}^R(\bX) = \sum_{x \in \bX \cap \Lambda_R} \phi(v-x)$, $H_v^{R,n}(\bX) = \sum_{x \in \bX \cap \Lambda_R^c \cap \Lambda_n} \phi(v-x)$ and $H_v^{>n}(\bX) =  \sum_{x \in \bX \cap \Lambda_n^c} \phi(v-x).$  Bound $$\E_{\mu_n}[e^{-H_v^R(\bX)} - e^{-H_v(\bX)}]  \leq \E_{\mu_n}[1 - e^{-H_v^{R,n}(\bX)}] + \E_{\mu_n}[1 - e^{-H_v^{>n}(\bX)}]\,.$$ 
		We note that the latter term is equal to $1 - \exp(-\sum_{x \in \bY \cap \Lambda_n^c} \phi(v - x ) )$ which tends to zero as $n \to \infty$ by temperedness of $\bY$.  The former term may be bounded by the GNZ equations \eqref{eq:GNZ} as in the proof of Lemma \ref{lem:bc-tempered} by $$\E_{\mu_n}[1 - e^{-H_v^{R,n}(\bX)}] \leq \lambda \int_{\Lambda_R^c}(1 - e^{-\phi(v - x)})\,dx$$
		which tends to zero as $R \to \infty$ by temperedness.  Applying the same argument for $\mu_n'$, we see that it is sufficient to show that for each fixed $R$ we have that $\lim_{n \to \infty} \E_{\mu_n} e^{-H_v^R(\bX)} = \lim_{n \to \infty} \E_{\mu_n'} e^{-H_v^R(\bX)}$.
		
		Since the function $\bX \mapsto  e^{-H_v^R(\bX)}$ is bounded, measurable, and local, to show convergence it will be sufficient to show that both $\mu_n$ and $\mu_n'$ converge to the same measure in the topology of local convergence (see Appendix B of \cite{jansen2019cluster} for formal definitions of the local topology).  To show convergence, we will show that every subsequence  has a \emph{further} subsequence that converges, and all of these subsequential limits are equal to the unique Gibbs measure with activity $\blambda$.  By Appendix B of \cite{jansen2019cluster}, every subsequence of $\{\mu_n\}$ contains a further subsequence that converges in the topology of local convergence.  Additionally, by \cite[Theorem B1 and Lemmas B1 and B2]{jansen2019cluster}, the subsequential limit is a Gibbs measure with activity $\bl$.  While Appendix B of \cite{jansen2019cluster} is written only for free boundary conditions, the proofs hold for tempered boundary conditions as well.  Since Gibbs measures are unique when $\blambda \leq \lambda < 1/(e C_\phi)$ (e.g.\ \cite[Theorem 2.1]{jansen2019cluster}) we have that every subsequence of $\{\mu_n\}$ has a further subsequence that converges to the unique Gibbs measure $\mu$, i.e.\ $\mu_n$ converges to $\mu$.  Since the same holds for $\mu_n'$, the proof is complete.
	\end{proof}
	We now prove Theorem \ref{thmAnalytic}.
	\begin{proof}[Proof of Theorem \ref{thmAnalytic}]
		We first prove analyticity of the pressure.  By Theorem \ref{thmZeros}, the finite volume pressure $p_V(\lambda) :=\frac{1}{V}\log Z_{\Lam_V}(\lambda) $ is an analytic function of $\lambda$ in  $L_0$ and uniformly bounded.  For $\lam \ge 0$, $p_V(\lam)$ converges as $V \to \infty$ (see, e.g.,\ \cite{ruelle1999statistical}).  Thus, by Vitali's Theorem (Theorem \ref{thmVitali}), $p_V(\lambda)$ converges uniformly to an analytic function in $L_0$.
		
		To prove uniqueness of the infinite volume Gibbs measure, we will use Lemma \ref{lem:bc-uniqueness}, and so we fix activity $\bl$ bounded by $\lambda$ along with tempered boundary conditions $\bY$ and $\bY'$ and a point $v$.  By Theorem~\ref{thmZeroFreeProof}, we may find a complex neighborhood $U$ of $[0,1]$ so that for $z \in U$, the function $z \mapsto \rho_{\Lambda_n,z \bl}^{\bY}(v)$ is well-defined, analytic and bounded in modulus by $C$.  To see that these functions are analytic, note that by \eqref{eq1ptDef}, they are ratios of analytic functions and by Theorem \ref{thmZeroFreeProof} the denominators are non-vanishing in $U$.  Let $f_n(z) = \rho_{\Lambda_n,z \bl}^{\bY}(v)$ and $g_n(z) =\rho_{\Lambda_n,z \bl}^{\bY'}(v)$ for $z \in U$.  For $z \in [0,1/e^2)$ we have $z \bl \in [0,1/(eC_\phi))$ and so by Lemma \ref{lem:small-lam-uniqueness} we have that $\lim_{n \to \infty}g_n(z) = \lim_{n \to \infty} f_n(z)$  for $z \in [0,1/e^2)$.  		
		By Vitali's Theorem, this implies that $f_n - g_n$ converges to an analytic function on $U$; since this limit is identically zero on the interval $[0,1/e^2)$ and analytic in $U$, we must have that the limit is $0$ on all of $U$ by the identity theorem for holomorphic functions.  By Lemma \ref{lem:bc-uniqueness}, this completes the proof of uniqueness.
	\end{proof}

	To prove  Theorem~\ref{corDensity} we use a lower bound on the density as a function of $\lam$ for $\lam$ positive, closely related to an inequality of Lieb~\cite[Eq. (1.19)]{lieb1963new}.  
	Fix the potential $\phi$, and for a region $\Lam$, let $\rho_{\Lam}(\lam) = \lam \frac{d}{d  \lambda}  \frac{1}{|\Lam|} \log Z_{\Lam} (\lam)$ be the finite volume density.
	\begin{lemma}
		\label{lemDensityLB}
		For any repulsive, tempered potential $\phi$, any region $\Lam$, and any $\lam \ge 0$, $\rho_{\Lam}(\lam) \ge \lam/(1+ \lam C_{\phi})$. 
	\end{lemma} 
	The proof below is a generalization of that of~\cite[Lemma 18]{helmuth2020correlation} for the case of hard spheres.
	\begin{proof}
		We denote by $\mathbf X$ the random point set in $\Lam$ drawn from the Gibbs measure $\mu_{\Lam,\lam}$, with density $e^{- U( \cdot)}$ against the Poisson process of intensity $\lam$ on $\Lam$.  The finite volume density is then $\frac{1}{|\Lam|} \E |\mathbf X|$.
		
		We will use a simple consequence of inclusion--exclusion, that for each $n$ and every sequence of numbers $x_1, \dots, x_n$, with $x_j \in [0,1]$ for all $j$, we have 
		$\prod_{j=1}^n (1 - x_j) \geq 1 - \sum_{j=1}^n x_j\,.$
		
		Now since $\rho_{\Lam}(\lam) = \lambda \frac{Z'(\lambda)}{|\Lam| Z(\lambda)}$ we have
		\begin{align*}
		\rho_{\Lam}(\lam)  &= \frac{\lambda}{|\Lam|Z_{\Lam}(\lambda)}\sum_{k \geq 0} \int_{\Lam^{k+1}} \frac{\lambda^{k}}{k!} \prod_{0 \leq i < j \leq k} e^{-\phi(x_i-x_j)}\,dx_1\,\ldots\,dx_k\,dx_0 \\
		&=\frac{\lambda}{|\Lam| Z_{\Lam}(\lambda)}  \int_{\Lam^k} \prod_{1 \leq i < j \leq k} e^{-\phi(x_i - x_j)} \frac{\lambda^k}{k!}  \left( \int_\Lam \prod_{j = 1}^k e^{-\phi(x_j - x_0)}  \,dx_0\right)  \,dx_1\,\ldots\,dx_k \\
		&= \frac{\lambda}{|\Lam|} \E \left[ \int_\Lam \prod_{y \in \bX} e^{-\phi(y - x)}\,dx \right] \\
		&= \frac{\lambda}{|\Lam|} \E \left[ \int_\Lam \prod_{y \in \bX} (1 -(1 -e^{-\phi(y - x)}))\,dx \right] \\
		&\geq \frac{\lambda}{|\Lam|}\E \left[ \int_\Lambda  \left (1 - \sum_{y \in \bX} (1 - e^{-\phi(y-x)}) \right)\,dx \right] \\ 
		&\geq \lambda\left(1 - \E\left[\frac{|\bX|}{|\Lam|} C_{\phi} \right]\right) \\
		&= \lambda - \lambda C_\phi \rho_{\Lam}(\lam) \,.
		\end{align*}
		Rearranging gives the lemma.
	\end{proof}

	With this we can prove Theorem~\ref{corDensity}. 
	\begin{proof}[Proof of Theorem~\ref{corDensity}]
		Fix $\lambda_0 \in [0,e/C_{\phi})$.  Then by Theorem \ref{thmZeroFreeProof}, there is an $\eps > 0$ so that for all  regions $\Lambda$ the pressure $\frac{1}{V}\log Z_{\Lambda_V}(\lambda)$ is analytic and bounded above for $\lambda  \in \mathcal{N}(\lambda_0,\eps)$; by the proof of Theorem \ref{thmAnalytic}, the finite volume pressure converges uniformly on compact subsets of $\mathcal{N}(\lambda_0,\eps)$ and so its derivative with respect to $\lam$ (multiplied by $\lam$) converges uniformly on compact subsets to the limit $\rho_{\phi}(\lambda)$.  In order to show that $p_\phi$ may be taken as an analytic function of $\rho_{\phi}$, it is sufficient to show that $\lambda$ may be taken as an analytic a function of $\rho_{\phi}$.  To show this, we will use the inverse function theorem for analytic functions (e.g.,\ \cite[Theorem $2.1.1$]{simon2015basic}).  Since $\rho_{\phi}$ is an analytic function of $\lambda$, it is sufficient to show that $\rho_{\phi}'(\lambda) \neq 0$ for $\lambda \in [0,\lambda_0]$.   This follows from, e.g., inequality (5) in~\cite{ginibre1967rigorous} which gives a uniform lower bound on the derivative of the finite volume density with respect  to $\lam$.  This shows that for $\lambda \in [0,\lambda_0]$ we have that $\rho_{\phi}'(\lambda) \neq 0$ and so $\rho_{\phi}^{-1}$ is analytic in a complex neighborhood of $[0,\rho_{\phi}(\lambda_0)]$.   By Lemma~\ref{lemDensityLB}, this interval contains the interval $\left [0, \frac{ \lam_0}{1+ \lam_0 C_{\phi}} \right]$.  Sending $\lam_0 \to e/C_{\phi}$ proves analyticity of $p_{\phi}$ in $\rho_{\phi}$ in the interval $\left [0, \frac{e}{1+e} \frac{1}{C_{\phi}} \right)$.
	\end{proof}
	
	To conclude we prove Corollary~\ref{corHSdensity}.
	\begin{proof}[Proof of Corollary~\ref{corHSdensity}]
		The first statement is an immediate consequence of Theorem~\ref{corDensity}.  The second statement follows by replacing the lower bound on the density from Lemma~\ref{lemDensityLB} with the bound, specific to the hard sphere model, from~\cite[Theorem 19]{helmuth2020correlation} (which in turn comes from~\cite{jenssen2019hard}), which says that for $\lam$ fixed, the packing density in dimension $d$ at fugacity $\lam$ is at least $(1+o_d(1)) W(\lam)  2^{-d}$, where $W(\cdot)$ is the W-Lambert function.  In particular, $W(e) =1$, which gives Corollary~\ref{corHSdensity}.
	\end{proof}

	\section*{Acknowledgements}
	We thank Tyler Helmuth,  Sabine Jansen, Steffen Betsch,  G{\"u}nter Last, and the anonymous referees for many helpful comments on the paper.  MM supported in part by NSF grant DMS-2137623. WP supported in part by NSF grants
	DMS-1847451 and CCF-1934915.

\end{document}